\documentclass[11pt,letterpaper]{article}

\usepackage[T1]{fontenc}
\usepackage[utf8]{inputenc}
\usepackage{mathpazo}
\usepackage{amsmath,amssymb,amsfonts,amsthm}
\usepackage{mathtools}  
\usepackage{xcolor}     
\usepackage{soul}       
\usepackage{microtype}  
\usepackage{float}
\usepackage{thmtools} 
\usepackage{thm-restate}
\usepackage{caption}
\usepackage[linesnumbered,boxed]{algorithm2e}

\usepackage[para]{footmisc}

\SetAlCapSkip{1em}
\SetAlCapNameFnt{\small}
\SetAlCapFnt{\small}

\definecolor{OliveGreen}{rgb}{0,0.6,0}

\usepackage{tikz}
\usetikzlibrary{arrows.meta, positioning}

\usepackage{hyperref}
\usepackage{cleveref}
\hypersetup{
  colorlinks=true,
  urlcolor=blue,
  linkcolor=blue,
  citecolor=OliveGreen
}

\usepackage[
letterpaper,
top=1in,
bottom=1in,
left=1in,
right=1in]{geometry}

\usepackage[scr=rsfso]{mathalfa}
\usepackage{bm}

\usepackage[backend=biber,
            style=trad-alpha,
            sorting=nty,
            doi=false,
            isbn=false,
            url=false,
            maxbibnames=20,
            maxcitenames=2]{biblatex}
\bibliography{bibliography}
\newbibmacro{string+doi}[1]{\iffieldundef{doi}{#1}{\href{https://dx.doi.org/\thefield{doi}}{#1}}}
\DeclareFieldFormat{title}{\usebibmacro{string+doi}{\mkbibemph{#1}}}
\DeclareFieldFormat[article]{title}{\usebibmacro{string+doi}{\mkbibquote{#1}}}
\DeclareFieldFormat[incollection]{title}{\usebibmacro{string+doi}{\mkbibquote{#1}}}
\DeclareFieldFormat[inproceedings]{title}{\usebibmacro{string+doi}{\mkbibquote{#1}}}
\renewbibmacro{in:}{}

\theoremstyle{definition} 
\newtheorem{theorem}{Theorem}[section]
\newtheorem*{theorem*}{Theorem}

\newtheorem*{proposition*}{Proposition}
\newtheorem{lemma}[theorem]{Lemma}
\newtheorem*{lemma*}{Lemma}
\newtheorem{corollary}[theorem]{Corollary}

\newtheorem*{conjecture*}{Conjecture}

\newtheorem*{fact*}{Fact}

\newtheorem*{hypothesis*}{Hypothesis}

\newtheorem*{claim*}{Claim}

\theoremstyle{definition}
\newtheorem{definition}[theorem]{Definition}

\newtheorem*{remark*}{Remark}

\newtheorem*{observation*}{Observation}

\newcommand{\Paren}[1]{\left(#1\right)}

\newcommand{\brac}[1]{[#1]}
\newcommand{\Brac}[1]{\left[#1\right]}

\newcommand{\Abs}[1]{\left\lvert#1\right\rvert}

\newcommand{\Norm}[1]{\left\lVert#1\right\rVert}

\newcommand{\Esymb}{\mathbb{E}}
\newcommand{\Psymb}{\mathbb{P}}

\DeclareMathOperator*{\E}{\Esymb}

\DeclareMathOperator*{\ProbOp}{\Psymb}
\renewcommand{\Pr}{\ProbOp}
\newcommand{\ex}[1]{\E\brac{#1}}
\newcommand{\Ex}[2][]{\E_{{#1}}\Brac{#2}}

\DeclareMathOperator{\Tr}{Tr}
\DeclareMathOperator{\poly}{poly}
\DeclareMathOperator{\negl}{negl}

\DeclareMathOperator{\polylog}{polylog}

\newcommand{\braket}[1]{\langle #1 \rangle}

\newcommand{\cC}{\mathcal C}
\newcommand{\cD}{\mathcal D}

\newcommand{\cF}{\mathcal F}

\newcommand{\cH}{\mathcal H}

\newcommand{\cR}{\mathcal R}
\newcommand{\cS}{\mathcal S}

\newcommand{\cU}{\mathcal U}

\newcommand{\bbI}{\mathbb I}

\newcommand*{\NC}{\mathsf{NC}}
\newcommand*{\AC}{\mathsf{AC}}
\newcommand*{\QNC}{\mathsf{QNC}}
\newcommand*{\QAC}{\mathsf{QAC}}
\newcommand*{\BPP}{\mathsf{BPP}}

\definecolor{BrickRed}{rgb}{0.8, 0.25, 0.33}
\definecolor{DarkBlue}{rgb}{0,0,0.8}
\definecolor{DarkOrange}{rgb}{0.8,0.4,0}

\newcommand{\Id}{\mathbb{I}}

\newcommand{\ket}[1]{|#1\rangle}
\newcommand{\bra}[1]{\langle#1|}

\renewcommand{\to}{\rightarrow}
\newcommand{\zo}{\{0,1\}}
\newcommand{\eps}{\varepsilon}
\newcommand{\haar}{\mathsf{Haar}}

\usepackage{authblk}
\title{Unconditional Pseudorandomness against Shallow Quantum Circuits}
\author[1]{Soumik Ghosh\thanks{\texttt{soumikghosh@uchicago.edu}}}
\author[2]{Sathyawageeswar Subramanian\thanks{\texttt{ss2310@cam.ac.uk}}}
\author[1]{Wei Zhan\thanks{\texttt{weizhan@uchicago.edu}}}
\affil[1]{University of Chicago}
\affil[2]{University of Cambridge}

\date{} 

\begin{document}
\maketitle

\begin{abstract}
    Quantum computational pseudorandomness has emerged as a fundamental notion that spans connections to complexity theory, cryptography and fundamental physics. However, all known constructions of efficient quantum-secure pseudorandom objects rely on complexity theoretic assumptions. 
    
    In this work, we establish the first \emph{unconditionally} secure efficient pseudorandom constructions against shallow-depth quantum circuit classes. We prove that:
    \begin{itemize}
    \item Any quantum state $2$-design yields unconditional pseudorandomness against both $\QNC^0$ circuits with arbitrarily many ancillae and $\AC^0\circ\QNC^0$ circuits with nearly linear ancillae. 
    \item Random phased subspace states, where the phases are picked using a $4$-wise independent function, are unconditionally pseudoentangled against the above circuit classes. 
    \item Any unitary $2$-design yields unconditionally secure parallel-query pseudorandom unitaries against geometrically local $\QNC^0$ adversaries, even with limited $\AC^0$ postprocessing.
    \end{itemize}
    Our indistinguishability results for $2$-designs stand in stark contrast to the standard setting of quantum pseudorandomness against \textsf{BQP} circuits, wherein they can be distinguishable from Haar random ensembles using more than two copies or queries. Our work demonstrates that quantum computational pseudorandomness can be achieved unconditionally for natural classes of restricted adversaries, opening new directions in quantum complexity theory.
\end{abstract}

\newpage
\section{Introduction}\label{sec: intro}
Randomness is fundamental to both classical and quantum computation, particularly in cryptography and algorithm design. However, true randomness is often scarce or computationally impractical. The theory of pseudorandomness studies deterministic objects that appear random to resource-bounded observers. For example, classical pseudorandom generators (PRGs) produce bit strings indistinguishable from random strings for computationally limited observers. A rich theory connects hardness to pseudorandomness for complexity classes such as $\BPP$ \cite{Nisan1991,ImpagliazzoNW94}.

Quantum pseudorandomness extends these ideas to quantum states and unitaries. Truly random quantum objects, described by the Haar measure over the unitary group \cite{mele2024introduction}, require exponential resources to generate. Two approaches have emerged for more efficient alternatives. First, the information-theoretic or statistical notion of efficient quantum $t$-designs: ensembles of states or unitaries that match the first $t$ moments of the Haar measure, and can be prepared or implemented efficiently \cite{Brando2016prl,Brando2016randomcircuits}. Second, the more recent computational approach: ensembles of quantum pseudorandom states (PRS) or unitaries (PRU) that appear Haar-random to computationally bounded quantum observers \cite{Ji2018}.

Classically, \emph{unconditional} pseudorandomness has been successfully constructed against several restricted computational models such as constant-depth circuits \cite{Nisan1991,Braverman08IndependenceFoolsAc0,DeETT10}, read-once branching programs \cite{Nisan92,ImpagliazzoNW94,NisanZ96}, and low-degree polynomials \cite{Viola09, BogdanovV10}. These results bypass the need for complexity theoretic assumptions, and have profound implications for cryptography \cite{Luby94book}, derandomization \cite{SaksZ99,Hoza21}, and lower bounds \cite{HealyVV06,GopalanMRTV12}.\footnote{Readers may refer to the survey by Hatami and Hoza \cite{HatamiH24} for a comprehensive review of the more recent developments on classical unconditional pseudorandomness.} 

In contrast, all existing quantum pseudorandom constructions target powerful adversaries such as polynomial-sized quantum circuits ($\mathsf{BQP}$), and rely on cryptographic assumptions such as the existence of quantum-secure one-way functions.  

\paragraph{Our contributions.}We establish the first unconditional efficient quantum pseudorandomness results against shallow-depth circuit classes. Such circuits model near-term quantum devices with limited coherence times and gate counts. We show that efficient pseudorandom objects, including PRS, pseudoentanglement, and PRU secure against parallel queries, can be constructed unconditionally for shallow quantum circuits. Our key insight is that due to the depth constraints, each output qubit of shallow quantum circuits locally depends only on a subset of input qubits, thus fundamentally limiting the ability of such circuits to distinguish certain structured quantum objects from Haar-random ones.

A notable aspect of our results is that the only property needed for our constructions is that of being an (approximate) $2$-design. \emph{A priori}, the design property only imposes conditions on the behavior of the object when few (in this case exactly two) copies of the objects are present, while pseudorandomness is a property that concerns an arbitrary polynomial number of copies. Rather surprisingly, our results bridge this gap, showing that when the power of the adversaries is restricted, \emph{information-theoretic} indistinguishability on two copies is strong enough to imply \emph{computational} indistinguishability on polynomially many copies.

Our work raises several open questions, ranging from constructing new classes of unconditionally pseudorandom objects against other shallow circuit classes, to applying these results to quantum cryptography and complexity theory. We discuss some of these directions in \Cref{sec: discussions}, providing new perspectives for analyzing near-term quantum devices.

\subsection{Main results}
In this article we construct unconditionally secure efficient pseudorandom objects against two important shallow-depth quantum circuit classes---$\QNC^0$ and $\AC^0\circ\QNC^0$. 

\paragraph{PRS from state designs.}We first demonstrate that unconditional pseudorandomness can be derived from state designs, whose security holds against circuits up to $\AC^0\circ\QNC^0$.

\begin{theorem}[Informal; See \Cref{thm:qnc_multiout}]
    Every $2$-design state ensemble is an unconditionally secure PRS against $\QNC^0$ circuits with arbitrarily many ancillae and almost linearly many bits of output.
\end{theorem}

\begin{theorem}[Informal; See \Cref{col:acqnc}]
    Every $2$-design state ensemble is an unconditionally secure PRS against $\AC^0\circ\QNC^0$ circuits with almost linearly many ancillae.
\end{theorem}

Observe that this stands in stark contrast to the case of \textsf{BQP} adversaries. Consider an ensemble that forms at most a $t$-design, for some $t = \poly(n)$. Then, there are cases---for instance, in the case of random stabilizer states---where a \textsf{BQP} adversary can distinguish between a state from this ensemble and a Haar random state \cite{Arunachalam2025,Gross2021} using more than $t$ copies.

Pseudoentanglement refers to the phenomenon whereby states having very low entanglement are indistinguishable from states having very low entanglement. We also prove that unconditional pseudoentanglement can be achieved against the above shallow quantum circuits.
\begin{theorem}[Informal; See \Cref{col:pre}]
    There exists efficiently constructible, unconditionally secure pseudoentanglement against $\QNC^0$ circuits, and against $\AC^0\circ\QNC^0$ circuits with poly-logarithmically many ancillae.
\end{theorem}

\paragraph{PRU from unitary designs.}Similarly, we also prove that unitary $t$-designs are unconditionally pseudorandom against geometrically local shallow quantum circuits when queried only in parallel (i.e. non-adaptively).

\begin{theorem}[Informal; See \Cref{thm:qncpru}]
    Every unitary $2$-design is an unconditionally non-adaptive secure PRU against $1$-dimensional geometrically local $\QNC^0$ circuits with arbitrarily many ancillae, and almost linear depth $1$-dimensional geometrically local $\QNC$ pre-processing. 
\end{theorem}

We also extend this PRU construction to $\QNC^0$ circuits with $\AC^0$ post-processing, with the caveat that it must be weakened slightly since we do not have a natural means to deal with ancillae in the post-processing phase.
\begin{theorem}[Informal; See \Cref{thm:acqncpru}]
    Every unitary $2$-design on $n$-qubits is an unconditionally non-adaptive secure PRU against a subclass of $\AC^0\circ\QNC^0$ circuits on a multiple of $n$ qubits, and the pre-processing $\QNC^0$ circuit before the queries is $1$-dimensional geometrically local. 
\end{theorem}

\paragraph{Proof techniques.} 
A recurring tool that we use in our proofs is that, over any subsystem, the reduced states of a Haar random state are close to maximally mixed with high probability (see \Cref{col:page_haar}).  Our PRU results require the analogue of this observation for the output of non-adaptive queries to a Haar random unitary (see \Cref{col:pru_haar}). For this, we bound the expected norm of partial traces of off-diagonal terms $\ket{v}\bra{w}$ conjugated by a Haar random unitary (see \Cref{lma:page_orth}).

To prove our results for $\QNC^0$ with $\AC^0$ post-processing, we observe that when the number of ancillae in the pre-processing $\QNC^0$ circuit is small, the resulting output distribution has high entropy, although the output distributions are no longer $k$-wise independent. To deal with this, we prove a generalization of Braverman's result \cite{Braverman08IndependenceFoolsAc0} that $\AC^0$ circuits cannot distinguish $k$-wise independent distributions from uniform, showing that $\AC^0$ circuits also fail to distinguish \emph{$k$-wise indistinguishable} distributions with high min-entropy (see \Cref{lma:fool2}). 

We achieve our pseudoentanglement construction using random phased subspace states, which are superpositions over the orthonormal basis vectors of a subspace with equal amplitudes and random $\pm 1$ phases. We show that such states, when instantiated with a $4$-wise independent function for the random phase, are indistinguishable from Haar random by shallow circuits, and have low von Neumann entropy across any cut (see \Cref{col:pre}).

We believe these technical developments may be of independent interest.

\subsection{Related work}
Quantum computational notions of pseudorandomness were introduced in \cite{Ji2018} and have been studied in a variety of recent works. For instance, many types of pseudoentangled states have been constructed against \textsf{BQP} distinguishers in recent work (for examples, see \cite{aaronson2023quantumpseudoentanglement, bouland2023publickeypseudoentanglementhardnesslearning, bouland2024hardnesslearninggroundstate, akers2024holographicpseudoentanglementcomplexityadscft, giurgicatiron2023pseudorandomnesssubsetstates, jeronimo2024pseudorandompseudoentangledstatessubset, feng2024dynamicspseudoentanglement}). These notions have found a wide range of applications, from cryptography \cite{ananth2022cryptographypseudorandomquantumstates, goulão2024pseudoentanglementnecessaryefipairs, ananth2023pseudorandomfunctionlikequantumstate} to physics \cite{yang2023complexitylearningpseudorandomdynamics, Gu_2024, Feng_2025, bouland2023publickeypseudoentanglementhardnesslearning, chakraborty2025fastcomputationaldeepthermalization}. A number of recent works have also considered the problem of constructing highly efficient pseudorandom unitaries that are implementable in extremely low depth \cite{schuster2025lowdepth,cui2025designs}. 

However, all these constructions rely on complexity theoretic assumptions to obtain pseudorandomness against polynomial-sized quantum circuits. Usually, the assumption has to do with the existence of quantum-secure one way functions \cite{zhandry}, based on computational hardness assumptions like the quantum hardness of the learning with errors (LWE) problem \cite{regev2024latticeslearningerrorsrandom}.

In contrast to this line of work, in our work the adversary is a class of shallow-depth quantum circuits against which we would like our pseudorandom constructions to be secure, and our contribution lies in showing that pseudorandomness against such circuit classes can be obtained without making any complexity theoretic assumptions.

\section{Preliminaries}
\label{sec:prelims}
We first define some commonly used notations. We use $[n]$ to denote the set $\{1,\ldots,n\}$. For two distributions $\cD$ and $\cD'$ over a set $X$ we use $\Abs{\cD-\cD'}_1$ to denote their total variation distance. We denote a random sample $x$ drawn according to $\cD$ by $x\sim\cD$, and we abuse the notation to denote $x$ drawn uniformly from a set $X$ by $x\sim X$. The identity operator on $n$ qubits is denoted as $\Id_n$.

We use $\Norm{\cdot}_p$ to denote the Schatten-$p$ norms of Hermitian operators. Specifically, $\Norm{\cdot}_1$, $\Norm{\cdot}_2$ and $\Norm{\cdot}_\infty$ respectively refers to the trace norm, Frobenius norm and operator norm.

We use the following shorthands for asymptotic growth: $\poly(n)=n^{O(1)}$, $\polylog(n)=\log^{O(1)}n$ and $\negl(n)=n^{-\omega(1)}.$ 

We assume the readers are familiar with the definitions of the following polynomial-sized circuit classes: $\QNC$ for quantum bounded fan-in circuits, $\AC$ for classical circuits with unbounded fan-in $\mathsf{AND}$ gates, and $\QAC$ for quantum circuits with unbounded size $\mathsf{CZ}$ gates (but without unbounded fan-out gates). We use $\QNC^0$, $\AC^0$ and $\QAC^0$ to denote their constant-depth subclasses respectively. Without loss of generality, we assume the bounded fan-in is at most $2$ (otherwise the constants in some of our results will be changed). For the purpose of this paper, we do not require the quantum circuits to compute cleanly: the ancillae could start with any specified state and also could end up in arbitrary states.

Following \cite{Slote24}, we also consider the hybrid circuits with quantum pre-processing and classical post-processing:
\begin{definition}
    For a class of classical circuits $\cF$ and a class of quantum circuits $\cC$, the circuit class $\cF\circ\cC$ consist of all circuits $F\circ C$ that are composed of a quantum circuit $C\in\cC$, followed by computational basis measurements on all output qubits of $C$, and then with some $F\in\cF$ applied on the measurement outcomes. The output distribution of $F\circ C$ with the input state $\rho$ is $F(C(\rho))$.
\end{definition}
The class that we are specifically interested in is $\AC^0\circ\QNC^0$, which is justified in \Cref{subsection: AC0}. It is shown in \cite{Slote24} that parity is hard to compute in this class, assuming either no ancillae or linear size of the $\AC^0$ circuit. It worth noting that we do not know yet whether $\AC^0\circ\QNC^0$ is comparable with $\QAC^0$.

\subsection*{Quantum Pseudorandom Primitives} Below we generalize the commonly used definitions of quantum pseudorandom primitives to those with respect to specific classes of adversaries, rather than simply polynomial-time adversaries. The state and unitary ensembles are all discrete distributions, which we denote by their supports for succinctness.

\begin{definition}
    The state ensemble $\{\ket{\psi_i}\}$ on $n$ qubits is a \emph{pseudorandom state ensemble (PRS)} against a class of quantum circuits $\cC$, if for the $n$-qubit Haar random state $\ket{\psi_\haar}$, every $t=\poly(n)$ and every circuit $C\in\cC$, we have
    \[\Abs{\Ex[i]{C\Paren{\ket{\psi_i}\bra{\psi_i}^{\otimes t}}}-\Ex{C\Paren{\ket{\psi_\haar}\bra{\psi_\haar}^{\otimes t}}}}_1=\negl(n).\]
    Here $C(\rho)$ represents the output distribution with input state $\rho$. 
\end{definition}

\begin{definition}
    We say two state ensembles $\{\ket{\psi_i}\}$ and $\{\ket{\psi_j}\}$ on $n$ qubits demonstrate \emph{pseudoentanglement} against a class of quantum circuits $\cC$, if every for every $t=\poly(n)$ and every circuit $C\in\cC$, we have
    \[\Abs{\Ex[i]{C\Paren{\ket{\psi_i}\bra{\psi_i}^{\otimes t}}}-\Ex[j]{C\Paren{\ket{\psi_j'}\bra{\psi_j'}^{\otimes t}}}}_1=\negl(n),\]
    while across the same bipartition or cut of the qubits, the expected entanglement entropies of $\{\ket{\psi_i}\}$ and $\{\ket{\psi_j'}\}$ are asymptotically different.
\end{definition}

\begin{definition}\label{def:PRU}
    The unitary ensemble $\{U_i\}$ on $n$ qubits is a \emph{pseudorandom unitary ensemble (PRU)} against a class of quantum circuits $\cC$, if for the $n$-qubit Haar random unitary $U_\haar$ and every circuit $C^U\in\cC^U$ on $t=\poly(n)$ input qubits, we have
    \begin{equation}\label{eq:prudef}
        \Abs{\Ex[i]{C^{U_i}(\ket{0^t}\bra{0^t})}-\Ex{C^{U_\haar}(\ket{0^t}\bra{0^t})}}_1=\negl(n).
    \end{equation}
    Here $C^U$ stands for a circuit $C$ which uses $U$ as oracle gates.

    In this work we are concerned with the notion of PRU when $U$ is guaranteed to be applied in parallel, that is, \eqref{eq:prudef} is only required to hold for circuits $C^U$ that apply all their $U$ gates in a single layer. In this case, we say $\{U_i\}$ on $n$ qubits is a \emph{parallel-query} (or \emph{non-adaptive-query}, as defined in \cite{metger2024pseudorandom}) \emph{secure PRU} against $\cC$. We refer to the part of the circuit $C^U$ before the layer of $U$ gates as pre-processing, and the part after as post-processing.
\end{definition}

We will also discuss some potential constructions of unconditional pseudorandom generators against shallow quantum circuits, defined as follows.

\begin{definition}
    The boolean function $G:\zo^\ell\to\zo^n$ is a \emph{$t$-copy pseudorandom generator (PRG)} against a class of quantum circuits $\cC$, if for every circuit $C\in\cC$, we have
    \[\Abs{\Ex[x\sim\zo^\ell]{C\Paren{\ket{G(x)}\bra{G(x)}^{\otimes t}}}-\Ex[x\sim\zo^n]{C\Paren{\ket{x}\bra{x}^{\otimes t}}}}_1=\negl(n).\]
    In particular, $G$ is a \emph{pseudorandom generator} against $\cC$ if it is a $t$-copy pseudorandom generator for every $t=\poly(n)$.
\end{definition}

\subsection*{State and Unitary Designs} Here we recall the definitions and properties of exact and approximate designs, which are statistical notions of pseudorandomness.

\begin{definition}[State $t$-design]
    The state ensemble $\{\ket{\psi_i}\}$ on $n$ qubits is a \emph{$t$-design}, if for the $n$-qubit Haar random unitary $\ket{\psi_\haar}$, we have
    \[\Ex[i]{\ket{\psi_i}\bra{\psi_i}^{\otimes t}}=\Ex{\ket{\psi_\haar}\bra{\psi_\haar}^{\otimes t}}.\]
    We say that $\{\ket{\psi_i}\}$ is an \emph{$\eps$-approximate $t$-design}, if instead we have
    \[\Norm{\Ex[i]{\ket{\psi_i}\bra{\psi_i}^{\otimes t}}-\Ex{\ket{\psi_\haar}\bra{\psi_\haar}^{\otimes t}}}_1\leq\eps.\]
\end{definition}
\begin{lemma}\label{lma:statedesign}
    Let $\{\ket{\psi_i}\}$ be an $\epsilon$-approximate state $2$-design on $n$ qubits, and let $B$ be a subsystem of containing $n-k$ qubits. We have
    \[\Ex[i]{\Norm{\Tr_B[\ket{\psi_i}\bra{\psi_i}]}_2^2}\leq \Ex{\Norm{\Tr_B[\ket{\psi_\haar}\bra{\psi_\haar}]}_2^2} + \eps.\]
\end{lemma}
\begin{proof}
Denote the complementary subsystem to $B$ by $A$, which contains $k$ qubits. Define the swap operator $R$ by the identity
\[\Tr[\Tr_B\rho_1\cdot\Tr_B\rho_2]
=\Tr\left[\left(\rho_1\otimes\rho_2\right)\cdot R\right],\]
where one can check that
\[R = \sum_{x,y\in\zo^k}\ket{x}\bra{y}_A\ket{y}\bra{x}_{A'}\otimes\Id_{BB'}.\]
Here $A'$ and $B'$ are identical copies of $A$ and $B$ respectively.

Since $R$ is a permutation matrix, the operator norm of $R$ is exactly $1$. As a result, by H\"{o}lder's inequality we have
\begin{align*}
    &\ \Ex[i]{\Norm{\Tr_B[\ket{\psi_i}\bra{\psi_i}]}_2^2}-\Ex{\Norm{\Tr_B[\ket{\psi_\haar}\bra{\psi_\haar}]}_2^2}  \\
    =&\ \Tr\left[\left(\Ex[i]{\ket{\psi_i}\bra{\psi_i}^{\otimes 2}}-\Ex{\ket{\psi_\haar}\bra{\psi_\haar}^{\otimes 2}}\right)\cdot R\right] \\
    =&\ \Norm{\Ex[i]{\ket{\psi_i}\bra{\psi_i}^{\otimes t}}-\Ex{\ket{\psi_\haar}\bra{\psi_\haar}^{\otimes t}}}_1\cdot\Norm{R}_\infty \leq \eps. \qedhere
\end{align*}
\end{proof}

\begin{definition}[Unitary $t$-design]\label{definition_tdesign}
    Let \(\mathcal{D}=\{U_i\}\) be an ensemble of \(n\)-qubit unitaries. The unitary ensemble $\mathcal{D}$  is a \emph{unitary $t$-design} if, for the $n$-qubit Haar‐random unitary $U_{\haar}$, we have
    \[\Ex[i]{U_i^{\otimes t}\otimes (U_i^\dagger)^{\otimes t}}=
    \Ex{U_{\haar}^{\otimes t}\otimes (U_{\haar}^\dagger)^{\otimes t}}.\]
    Define the $t$-th moment channels as
    \[\Phi_\cD^{(t)}(\rho)=\Ex[i]{U_i^{\otimes t}\cdot\rho\cdot(U_i^\dagger)^{\otimes t}},\qquad
    \Phi_\haar^{(t)}(\rho)=\Ex{U_{\haar}^{\otimes t}\cdot\rho\cdot(U_{\haar}^\dagger)^{\otimes t}},\]
    we say that $\mathcal{D}$ is an \emph{$\eps$-approximate unitary $t$-design}, if for all operators $\rho$ with $\Norm{\rho}_1\leq 1$ we have
    \[\Norm{\Phi_\cD^{(t)}(\rho)-\Phi_\haar^{(t)}(\rho)}_1\leq\eps.\]
\end{definition}

\begin{lemma}
\label{lma:approxdesign}
    Let $\{U_i\}$ be an $\epsilon$-approximate unitary $2$-design on $n$ qubits, and let $B$ be a subsystem of the $n$ qubits. For every two $n$-qubit states $\ket{v}$ and $\ket{w}$, we have
    \[
    \Ex[i]{\Norm{\Tr_B[U_i\ket{v}\bra{w}U_i^\dagger]}_2^2}\leq \Ex{\Norm{\Tr_B[U_\haar\ket{v}\bra{w}U_\haar^\dagger]}_2^2} + \eps.
    \]
\end{lemma}
\begin{proof}
Denote the complementary subsystem to $B$ by $A$, and assume that $A$ contains $k$ qubits. Define the operator $R$ the same way as the proof above for \Cref{lma:statedesign}, and we have
\[\Norm{\Tr_B(U|v\rangle\langle w|U^\dagger)}_2^2 
    =\Tr\left[(U\ket{v}\bra{w}U^\dagger)^{\otimes 2}\cdot R\right].\]
Similarly, since $\Norm{R}_\infty=1$ and $\Norm{\ket{v}\bra{w}^{\otimes2}}_1=1$, by H\"{o}lder's inequality we have
\begin{align*}
    &\ \Ex[i]{\Norm{\Tr_B[U_i\ket{v}\bra{w}U_i^\dagger]}_2^2}-\Ex{\Norm{\Tr_B[U_\haar\ket{v}\bra{w}U_\haar^\dagger]}_2^2}  \\
    =&\ \Tr\left[\left(\Phi_\cD^{(t)}(\ket{v}\bra{w}^{\otimes2})-\Phi_\haar^{(t)}(\ket{v}\bra{w}^{\otimes2})\right)\cdot R\right] \\
    =&\ \Norm{\Phi_\cD^{(t)}(\ket{v}\bra{w}^{\otimes2})-\Phi_\haar^{(t)}(\ket{v}\bra{w}^{\otimes2})}_1\cdot\Norm{R}_\infty \leq \eps. \qedhere
\end{align*}
\end{proof}

\subsection*{Schmidt Decomposition} 
We will need several facts about the Schmidt decomposition (listed below), whose proofs can be found in e.g. \cite{NielsenChuang}.
\begin{definition}
    Let $\cH_1,\cH_2$ be two Hilbert spaces, and let $x\in\cH_1\otimes\cH_2$. If we write
    \begin{equation}\label{eq:tensorprod}
        x=\sum_{i=1}^r\alpha_i\cdot v_i\otimes w_i,
    \end{equation}
    where $\alpha_i\in\mathbb{C}$, $v_i\in\cH_1$ and $w_i\in\cH_2$, we call \eqref{eq:tensorprod} a \emph{tensor product decomposition} of $x$. Furthermore, if both $\{v_i\}$ and $\{w_i\}$ are orthonormal and each $\alpha_i$ is non-zero, we call \eqref{eq:tensorprod} a \emph{Schmidt decomposition} and $r$ the \emph{Schmidt rank} of $x$ with respect to $\cH_1$ and $\cH_2$.
\end{definition}
\begin{lemma}
\label{schmidt rank}
    Let $\cH_1,\cH_2$ be two Hilbert spaces, and let $x\in\cH_1\otimes\cH_2$. Then:
    \begin{itemize}
        \item In any tensor product decomposition of $x$ as in \eqref{eq:tensorprod}, the number of terms $r$ is at least the Schmidt rank of $x$;
        \item Let $\{v_i\}$ and $\{w_j\}$ be two orthonormal basis for $\cH_1$ and $\cH_2$, respectively. If we write
        \[x=\sum_{i,j}\alpha_{ij}\cdot v_i\otimes w_j,\]
        then the Schmidt rank of $x$ is exactly the rank of the matrix with entry $\alpha_{ij}$ at the $i$-th row and $j$-th column.
        \item The von Neumann entanglement entropy of $x$, with respect to the subsystems $\cH_1$ and $\cH_2$, is at most $\log_2 r$ where $r$ is the Schmidt rank of $x$.
    \end{itemize}
\end{lemma}

\section{Unconditional pseudorandomness from 2-designs}
\label{section: pseudorandomness from designs}
In this section, we focus on state designs which exploit the locality properties of shallow circuits in order to achieve unconditional pseudorandomness. At a high level, this resembles a quantum analog of small bias distributions (e.g.\ \cite{NaorN93}), which can fool low-degree polynomials. We will begin with some facts about Haar random states, which relate the size of the subsystem with entanglement entropy, and allow us to approximate small subsystems with maximally mixed states.

\begin{lemma}[\cite{lubkin1978entropy,liu2018entanglement}]\label{lma:page}
    Let $\ket{\psi_\haar}$ be an $n$-qubit Haar random state, and let $\rho_A$ be the reduced density matrix of $\ket{\psi_\haar}\bra{\psi_\haar}$ on the subsystem $A$ with $|A|=k$ qubits. Then
    \[\ex{\Tr(\rho_A^2)}=\frac{2^k+2^{n-k}}{2^n+1}.\]
\end{lemma}

\begin{corollary}\label{col:page_haar}
    Let $\ket{\psi_\haar}$ be an $n$-qubit Haar random state. For any $t\geq 1$, let $\rho_A$ be the reduced density matrix of $\ket{\psi_\haar}\bra{\psi_\haar}^{\otimes t}$ over a subsystem $A$. Then for every $\delta>0$, with probability at least $1-n^{O(k)}\cdot2^{-n/2}\cdot\delta^{-1}$ over $\ket{\psi_\haar}$, it holds for all $A$ with $|A|=k$ qubits that
    \[\Norm{\rho_A-\frac{1}{2^k}\Id_k}_1\leq \delta.\]
\end{corollary}
\begin{proof}
    First consider the case when $A$ is fully contained in one copy of the Haar random state. In this case from \Cref{lma:page} we have
    \begin{align*}
       \Ex{\Norm{\rho_A-\frac{1}{2^k}\Id_k}_1}
       &\leq 2^{k/2}\cdot\Ex{\Norm{\rho_A-\frac{1}{2^k}\Id_k}_2}\\
       &\leq2^{k/2}\cdot\Ex{\Norm{\rho_A-\frac{1}{2^k}\Id_k}_2^2}^{1/2} \\
       &=2^{k/2}\cdot\Ex{\Tr(\rho_A^2)-2^{-k}}^{1/2}  \\
       &=2^{k/2}\cdot\Paren{\frac{2^k-2^{-k}}{2^n+1}}^{1/2} \\
       &\leq 2^{k-n/2}.
    \end{align*}
    By Markov's inequality, we know that $\Norm{\rho_A-\frac{1}{2^k}\bbI_k}_1\leq\delta/k$ holds with probability at least $1-k\cdot2^{k-n/2}\cdot\delta^{-1}$. By a union bound, this holds for all $A$ with $|A|\leq k$ with probability at least $1-n^{O(k)}\cdot2^{-n/2}\cdot\delta^{-1}$.
    
    When $A$ consists qubits from at most $k$ different copies, we denote the subsystems as $A=A_1\sqcup A_2\sqcup\cdots$ with $|A_i|=k_i$. Since the copies are unentangled with each other, we have $\rho_A=\rho_{A_1}\otimes\rho_{A_2}\otimes\cdots$, and thus
    \begin{align*}
        \Norm{\rho_A-\frac{1}{2^k}\Id_k}_1
        &\leq\sum_i\Norm{\rho_{A_1}\otimes\cdots\otimes\rho_{A_i}-\rho_{A_1}\otimes\cdots\otimes\rho_{A_{i-1}}\otimes\frac{1}{2^{k_i}}\Id_{k_i}}_1 \\
        &=\sum_i\Norm{\rho_{A_i}-\frac{1}{2^{k_i}}\Id_{k_i}}_1 \\
        &\leq \delta
    \end{align*}
    with probability at least $1-n^{O(k)}\cdot2^{-n/2}\cdot\delta^{-1}$.
\end{proof}

Notice that the proof of \Cref{col:page_haar} only uses the second moment properties of $\ket{\psi_\haar}$, and therefore the conclusions immediately hold for approximate 2-designs with negligible error as well.
\begin{corollary}\label{col:page_design}
    Let $\{\ket{\psi_i}\}$ be an $\eps$-approximate $2$-design on $n$ qubits. For any $t\geq 1$, let $\rho_A$ be the reduced density matrix of $\ket{\psi_i}\bra{\psi_i}^{\otimes t}$ over a subsystem $A$. Then for every $\delta>0$, with probability at least $1-n^{O(k)}\cdot(\eps+2^{-n})^{1/2}\cdot\delta^{-1}$ over $i$, it holds for all $A$ with $|A|=k$ qubits that
    \[\Norm{\rho_A-\frac{1}{2^k}\Id_k}_1\leq\delta.\]
\end{corollary}
\begin{proof}
    By \Cref{lma:statedesign}, the approximate design property implies that
    \[\Ex{\Tr(\rho_A^2)-2^{-k}}^{1/2}\leq (\eps+2^{k-n})^{1/2},\]
    and thus
    \[\Ex{\Norm{\rho_A-\frac{1}{2^k}\Id_k}_1}\leq 2^k\cdot(\eps+2^{-n})^{1/2}.\]
    The rest of the proof follows the same arguments from \Cref{col:page_haar}.
\end{proof}

\subsection{Pseudorandomness against \texorpdfstring{$\QNC^0$}{QNC0}}
As a warm-up, we will use \Cref{col:page_haar} and \Cref{col:page_design} to show that any $2$-design is indistinguishable to a Haar random state, with respect to any $\QNC^{0}$ distinguisher. 

\begin{theorem}
\label{thm:prs}
Let $\{\ket{\psi_i}\}$ be an $\eps$-approximate $2$-design on $n$ qubits for some $\eps = \negl(n)$. Then $\{\ket{\psi_i}\}$ is a PRS against $\QNC$ circuits with depth
\[d=\min(\log\log(1/\eps),\log n)-\log\log n-\omega(1)\]
and a single-bit output. In particular, when $\eps\leq 2^{-\Omega(n)}$, $\{\ket{\psi_i}\}$ is PRS against $\QNC$ circuits up to depth $d=\log n-\log\log n-\omega(1)$.
\end{theorem}

\begin{proof}
The output bit of the depth-$d$ $\QNC$ circuit depends on at most $k=2^d$ input qubits and ancillae. Let $A$ be the subsystem of the input state on these qubits, and let $m$ be the number of ancillae touched. Denote the reduced density matrix, over the subsystem $A$, of the Haar random state to be $\rho^{\haar}_A$ and that of the state picked from $\{\ket{\psi_i}\}$ to be $\rho_A$. Then for the channel $\Phi_\cC$ that maps the subsystem $A$ to the output qubit, we have
\[\Norm{\Phi_\cC\left(\rho^{\haar}_{A}\right) - \Phi_\cC\left(\rho_A\right)}_1 
\leq \Norm{\rho^{\haar}_A - \rho_{A}}_1 =\negl(n),
\]
with probability at least $1 - \negl(n)$ over the choice of the states. This follows from \Cref{col:page_haar} and \Cref{col:page_design}, and the observation that
\[n^{O(k)}\cdot (\eps+2^{-n})^{1/2}\leq\negl(n)\]
is equivalent to 
\[O(k)\leq\frac{\log(1/(\eps+2^{-n}))}{2\log n}-\omega(1),\]
which is satisfied for every $d=\log k$ such that
\[d\leq \min(\log\log(1/\eps),\log n)-\log\log n-\omega(1). \qedhere\]
\end{proof}
The above theorem can be strengthened to work for the case when multiple output qubits are measured.
\begin{corollary}\label{thm:qnc_multiout}
    Let $\{\ket{\psi_i}\}$ be an $\eps$-approximate $2$-design on $n$ qubits for some $\eps=2^{-\Omega(n)}$. Then $\ket{\psi_i}$ is a PRS against $\QNC$ circuits with depth $d$ and $k\leq 2^{-d}\cdot o(n/\log n)$ bits of output.
\end{corollary}

\begin{proof}
The backward lightcone of the $k$ output qubits is of at most 
\[
k \cdot 2^d = o(n/\log n)
\]
in size. The proof then follows from the same argument as \Cref{thm:prs}.
\end{proof}

\subsection{Pseudorandomness against \texorpdfstring{$\AC^0\circ\QNC^0$}{AC0\bullet QNC0}}
\label{subsection: AC0}
The lightcone argument in the previous section renders most parts of a $\QNC^0$ circuit and most input qubits unrelated. Naturally, it is more desirable to prove the statement against $\QNC^0$ circuits where all the output qubits are measured, that is, the output distribution is indistinguishable in total variation distance between $t$-designs and Haar random states. However, our example below shows that this is too much to ask for, even when the $\QNC^0$ circuit does nothing and the input states get measured immediately.
\begin{definition}[Random phased subspace states]\label{def:subspaces}
    A $d$-dimensional \emph{random phased subspace state} on $n$ qubits is the following state:
    \[\ket{\psi_{S,f}}=\frac{1}{2^{d/2}}\sum_{x\in S}(-1)^{f(x)}\ket{x},\]
    where $S\in\mathbb{F}_2^n$ is a random $d$-dimensional linear subspace, and $f:S\to\zo$ is a random function.
\end{definition}
\begin{theorem}\label{thm:qnc0cte}
    For every $n>d>t$, $\{\ket{\psi_{S,f}}\}$ is an $O(2^{t-d})$-approximate $t$-design, but $\ket{\psi_{S,f}}^{\otimes(d+1)}$ and $\ket{\psi_\haar}^{\otimes(d+1)}$ are distinguishable when measured in the computational basis, with total variation distance $1-O(2^d)/2^n$.
\end{theorem}
\begin{proof}
    The proof that $\{\ket{\psi_{S,f}}\}$ is an approximate design is deferred to \Cref{sec:appendixA}. To distinguish $d+1$ copies of $\{\ket{\psi_{S,f}}\}$ from Haar random, notice that the measurement outcome in each copy is a random $x\in S$, and thus the $d+1$ outcomes must be linearly dependent. On the other hand, the measurement outcomes from $d+1$ copies of a Haar random state, when all distinct, form a random $(d+1)$-element subset of $\zo^n$ because of symmetry, and therefore are linearly dependent with probability at most 
    \[\frac{1}{2^n}+\frac{1}{2^{n-1}}+\cdots+\frac{1}{2^{n-d}}\leq \frac{1}{2^{n-d-1}}.\]
    We also know that the collision probability for a Haar random state is $2/(2^n+1)$ \cite{dalzell2022random}, and thus the probability for the outcomes not being all distinct is at most $d(d+1)/2^n$ by the union bound.
    As a result, the total variation distance between the measurement outcomes of $\ket{\psi_{S,f}}^{\otimes(d+1)}$ and $\ket{\psi_\haar}^{\otimes(d+1)}$ is at least $1-1/2^{n-d-1}-d(d+1)/2^n=1-O(2^d)/2^n$.
\end{proof}

Notice that not only the distinguisher in \Cref{thm:qnc0cte} applies no quantum gates, the classical post-processing on the measurement outcomes are also quite simple. It checks linear dependence whose complexity is captured by $\mathsf{DET}$, a complexity class between $\NC^1$ and $\NC^2$ containing all problems reducible to determinant. This motivates us to examine the case when the classical post-processing is restricted to some provably weaker class than $\mathsf{DET}$. It turns out that for $\AC^0$, $2$-designs are indeed still pseudorandom in this case. We crucially make use of the follow result, first proved by Braverman \cite{Braverman08IndependenceFoolsAc0}, subsequentially improved by \cite{Tal17tight} and Harsha and Srinivasan \cite{HarshaS19}, that almost $k$-wise uniform distributions fools $\AC^0$:
\begin{lemma}\label{lma:fool}
    For every $\delta$-almost $k$-wise independent distribution on $m$ bits, any $\AC$ circuits with size $s$ and depth $d$ cannot distinguish it from the uniform distribution with advantage $\eps+m^k\delta$, for certain $k=(\log s)^{O(d)}\cdot\log(1/\eps)$.
\end{lemma}
We start out with the simple case, when no ancilla is allowed for the $\QNC^0$ circuit.
\begin{theorem}\label{thm:acqnc}
    Let $\{\ket{\psi_i}\}$ be an $\eps$-approximate $2$-design on $n$ qubits for some $\eps=2^{-\log^{\omega(1)}n}$. Then $\ket{\psi_i}$ is a PRS against $\AC^0\circ\QNC^0$ circuits with no ancilla.
\end{theorem}
\begin{proof}
    Suppose the circuit has size $s$ and depth $d$. Fix some $k=(\log s)^{O(d)}\cdot\log^2 n=\log^{O(1)} n$ according to \Cref{lma:fool}.

    The output of the $\QNC^0$ circuit, when all qubits are measured, is a distribution $\mathcal{D}_{\ket{\psi}}$ over $m=\poly(n)$ bits that depends on the input state $\ket{\psi}$. By \Cref{col:page_design}, for both $\mathcal{D}_{\ket{\psi_i}}$ and $\mathcal{D}_{\ket{\psi_\haar}}$, with probability $1-\negl(n)$, the marginal distribution on every $k$ bits are $\delta$-indistinguishable from the case when the input states are maximally mixed, for every $\delta>0$ such that $n^{O(k)}\cdot(\eps+2^{-n})^{1/2}\cdot\delta^{-1}=\negl(n)$.
    
    Since there is no ancilla, the output distribution when the inputs are maximally mixed is the uniform distribution. Therefore both $\mathcal{D}_{\ket{\psi_i}}$ and $\mathcal{D}_{\ket{\psi_\haar}}$ are $\delta$-almost $k$-wise independent. By \Cref{lma:fool} both distributions are indistinguishable from the uniform distribution against the $\AC^0$ post-processing, as long as $m^k\delta=n^{O(k)}\delta$ is negligible. This is satisfied by our choice of $\eps$, as
    \[n^{O(k)}\cdot (\eps+2^{-n})^{1/2}=2^{\log^{O(1)}n-\log^{\omega(1)}n}=\negl(n).\qedhere\]
\end{proof}

Notice that for exact $2$-designs, that is when $\eps=0$, we only need $k=o(n/\log n)$, and thus \Cref{thm:acqnc} can be strengthened to work against $\AC\circ\QNC$ circuits of polynomial size, $\QNC$ depth up to $o(\log n)$ and $\AC$ depth up to $o(\log n/\log\log n)$.

The situation becomes more complicated when ancillae are allowed. In this case, although $\mathcal{D}_{\ket{\psi_i}}$ and $\mathcal{D}_{\ket{\psi_\haar}}$ are still almost $k$-wise indistinguishable (that the marginal distribution on every $k$ bits are close in total variation distance), they are no longer $k$-wise independent and are not guaranteed to fool $\AC^0$ circuits when $k$ is small \cite{BogdanovIVW16}. This happens because we have no knowledge of the output distribution of the $\QNC^0$ circuit with ancillae even when the inputs are maximally mixed. The saving grace is that, when the number of ancillae is small, the output distribution has high entropy and we can modify Braverman's proof in \cite{Braverman08IndependenceFoolsAc0} to suit such distributions:
\begin{lemma}\label{lma:fool2}
    For every two $\delta$-almost $k$-wise indistinguishable distributions $\cD_1,\cD_2$ on $m$ bits, such that $\cD_1$ has min-entropy at least $m-r$, any $\AC$ circuits with size $s$ and depth $d$ cannot distinguish the two distributions with advantage $\eps+4m^k\delta$, for certain $k=(\log s)^{O(d)}\cdot(r+\log(1/\eps))$.
\end{lemma}
\begin{proof}
    We first consider the case when $\delta=0$. Suppose the $\AC^0$ circuit with size $s$ and depth $d$ computes a boolean function $F$, Braverman showed that \cite[Lemma 11]{Braverman08IndependenceFoolsAc0} there exists a boolean function $F'$ and polynomial $f'$ of degree $k=(\log s)^{O(d)}\cdot\log(1/\eps')$, such that:
    \begin{itemize}
        \item $\Pr_{\cD_2'}[F\neq F']<\eps'$,
        \item $\Pr_{\cU}[F\neq F']<\eps'$,
        \item $F'\geq f'$ on $\zo^m$ and $\Ex[\cU]{F'-f'}<\eps'$.
    \end{itemize}
    Here $\cU$ stands for the uniform distribution over $\zo^m$. Since $\cD_1$ has min-entropy at least $m-r$, we have $\Pr_{\cD_1}[F\neq F']<2^r\eps$ and $\Ex[\cD_1]{F'-f'}<2^r\eps$. As a result,
    \begin{align*}
        \Ex[\cD_2]{F} 
        &> \Ex[\cD_2]{F'}-\eps'\geq\Ex[\cD_2]{f'}-\eps' \\
        &=\Ex[\cD_1]{f'}-\eps'>\Ex[\cD_1]{F'}-(2^r+1)\eps' \\
        &>\Ex[\cD_1]{F}-(2^{r+1}+1)\eps'.
    \end{align*}
    The bound in the reverse direction also holds by considering $1-F$. Taking $\eps'=\eps/(2^{r+1}+1)$ gives us $\Abs{\Ex[\cD_1]{F}-\Ex[\cD_2]{F}}<\eps$.

    For $\delta>0$, we can assume that $m^k\delta<1$, as otherwise the lemma is trivial. By \cite{BogdanovIVW16}, there exists two $k$-wise indistinguishable distributions $\cD_1',\cD_2'$ on $m$ bits such that $\Abs{\cD_1-\cD_1'}_1\leq 2m^k\delta$ and $\Abs{\cD_2-\cD_2'}_1\leq 2m^k\delta$. Furthermore, the construction ensures that $\Norm{\cD_1'}_\infty\leq\Norm{\cD_1}_\infty+2m^k\delta\cdot 2^{-m}$, and thus $\cD_1'$ still has min-entropy at least $m-O(r)$.  Applying the previous claim on $\cD_1',\cD_2'$ we have $\Abs{\Ex[\cD_1']{F}-\Ex[\cD_2']{F}}<\eps$, and thus $\Abs{\Ex[\cD_1]{F}-\Ex[\cD_2]{F}}<\eps+4m^k\delta$.
\end{proof}

\begin{theorem}\label{thm:acqnc2}
    Let $\{\ket{\psi_i}\}$ be an $\eps$-approximate $2$-design on $n$ qubits for some $\eps=2^{-\log^{\omega(1)}n}$. Then $\ket{\psi_i}$ is a PRS against $\AC^0\circ\QNC^0$ circuits with $a=\polylog(n)$ ancillae.
\end{theorem}
\begin{proof}
    The forward lightcone of the ancillae touches at most $r=2^da=\log^{O(1)} n$ output qubits. We called the these qubits \emph{corrupted}, denoted as the subsystem $R$. We define $\cR$ as the distribution over $\zo^r$ following the measurement outcome on the corrupted qubits when the input states are maximally mixed.
    
    Similar to the proof of \Cref{thm:acqnc}, suppose the circuit has size $s$ and depth $d$, and we fix some $k=(\log s)^{O(d)}\cdot(r+\log^2 n)=\log^{O(1)} n$ according to \Cref{lma:fool2}. Denote the output distribution of the $\QNC^0$ circuit as $\mathcal{D}_{\ket{\psi}}$ with the input state $\ket{\psi}$. Let $\delta>0$ satisfies $n^{O(k)}\cdot(\eps+2^{-n})^{1/2}\cdot\delta^{-1}=\negl(n)$. We claim that with probability $1-\negl(n)$, both $\mathcal{D}_{\ket{\psi_i}}$ and $\mathcal{D}_{\ket{\psi_\haar}}$ are $\delta$-almost $k$-wise indistinguishable from $\cU\otimes\cR$, where $\cU$ is the uniform distribution over the uncorrupted output bits.

    \begin{figure}[h]
        \centering
        \includegraphics[width = .45\textwidth]{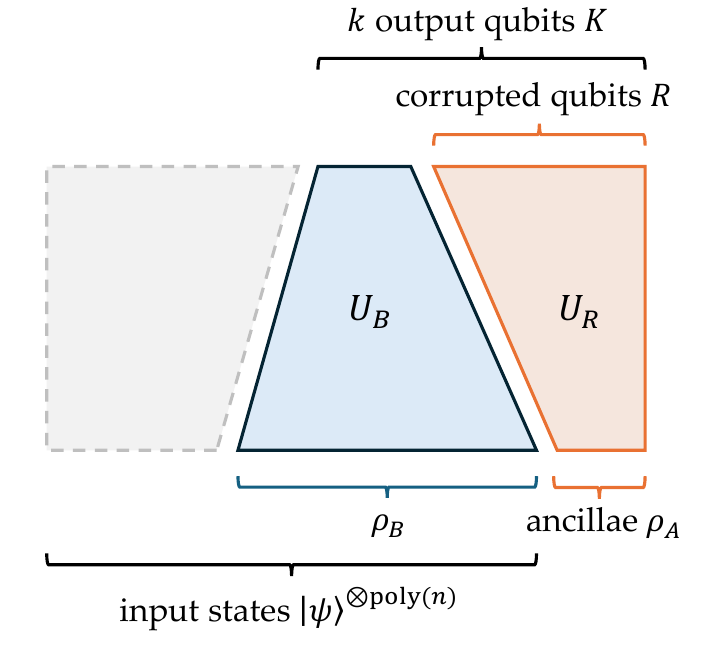}
        \caption{An illustration of the proof of \Cref{thm:acqnc2}, with partial systems and unitary lightcones in an $\QNC^0$ circuit.}
    \end{figure}
    
    To prove the claim, let $A$ be the ancilla qubits, and $\rho_A$ be the initial state of the ancillae. Let the unitary operator $U_R$ consist of all gates in the $\QNC^0$ circuit that belongs to the lightcone of the ancillae. Fixing any $k$ output qubits $K$, we extending $K$ to $K\cup R$ so that it contains all the corrupted qubits, and denote $U_B$ as the unitary operator that consists of all gates in the backward lightcone of $K$ but outside $U_R$. Notice that we can think of $U_R$ as being applied after $U_B$. At the input, the backward lightcone of $K$ touches the set of non-ancillae qubits $B$ with $|B|=b$, and we let $\rho_B$ be the state of $B$ when the input states are copies of $\ket{\psi}$.
    
    Now the output state on $K$ is a partial trace of
    \begin{equation}\label{eq:ac0}
        (\Id_{A\cup B\setminus R}\otimes U_R)(U_B\rho_BU_B^\dagger\otimes\rho_A)(\Id_{A\cup B\setminus R}\otimes U_R^\dagger),
    \end{equation}
    and it suffices to show that no matter if $\ket{\phi}$ is drawn from $\{\ket{\psi_i}\}$ or $\ket{\psi_\haar}$, with high probability the above state is close to the case when $\rho_B$ is maximally mixed. Indeed, by \Cref{col:page_haar} and \Cref{col:page_design}, in both former cases with probability $1-\negl(n)$ we have $\Norm{\rho_B-2^{-(k-a)}\Id_{B}}_1\leq\delta$. This means that the state in \eqref{eq:ac0} if $\delta$-close in trace distance to
    \[
        (\Id_{A\cup B\setminus R}\otimes U_R)(2^{-b}\Id_{B}\otimes\rho_A)(\Id_{A\cup B\setminus R}\otimes U_R^\dagger) 
        =\ 2^{-(a+b-r)}\Id_{A\cup B\setminus R}\otimes2^{-(r-a)}U_R(\Id_{R\setminus A}\otimes\rho_A)U_R^\dagger.
    \]
    Notice that $2^{-(r-a)}U_R(\Id_{R\setminus A}\otimes\rho_A)U_R^\dagger$ is exactly the output state on $R$ when the input states are maximally mixed, and therefore the measurement outcome has the distribution $\cR$. Meanwhile $2^{-(a+b-r)}\Id_{A\cup B\setminus R}$ is maximally mixed and will be measured to the uniform distribution. Thus we conclude that the output distributions on the $k$ bits are $\delta$-close to $\cU\otimes\cR$, for both $\mathcal{D}_{\ket{\psi_i}}$ and $\mathcal{D}_{\ket{\psi_\haar}}$.

    Now we use the fact that $\cU\otimes\cR$, as a distribution over $\zo^m$, has min-entropy at least $m-r$. By \Cref{lma:fool2}, both $\mathcal{D}_{\ket{\psi_i}}$ and $\mathcal{D}_{\ket{\psi_\haar}}$ are indistinguishable from $\cU\otimes\cR$ against the $\AC^0$ post-processing, as long as $m^k\delta=n^{O(k)}\delta$ is negligible. Similar to the proof of \Cref{thm:acqnc}, this is satisfied by our choice of $\eps$.
\end{proof}
Similar to the case of \Cref{thm:acqnc}, for exact $2$-designs, \Cref{thm:acqnc2} can be strengthened to work against $\AC\circ\QNC$ circuits of polynomial size, $\QNC$ depth up to $o(\log n)$ and $\AC$ depth up to $o(\log n/\log\log n)$. On the other hand, for constant depth circuits the number of ancillae can be strengthened close to linear.
\begin{corollary}\label{col:acqnc}
    Let $\{\ket{\psi_i}\}$ be an exact $2$-design on $n$ qubits. Then $\ket{\psi_i}$ is a PRS against $\AC^0\circ\QNC^0$ circuits with $a=n/\log^{\omega(1)} n$ ancillae.
\end{corollary}
\begin{proof}
    In the proof of \Cref{thm:acqnc2}, when $\eps=0$ we can take $\delta=2^{-n/4}$, and thus to have $m^k\delta=\negl(n)$ for any $m=\poly(n)$ we only need $k=o(n/\log n)$. When $d=O(1)$, this is satisfies by any $a=2^{-d}r=n/\log^{\omega(1)} n$.
\end{proof}

\subsection{Pseudoentanglement against \texorpdfstring{$\AC^0\circ\QNC^0$}{AC0\bullet QNC0}}
In contrast to prior works that constructed pseudoentanglement from quantum secure one-way functions, here we prove that unconditional pseudoentanglement is possible against shallow circuits, using random phased subspace states (see \Cref{def:subspaces}). We will show that such states form good enough approximate $t$-designs, even when the phases are picked using a $2t$-wise independent function, and thus yield pseudorandomness by our previous results.

\begin{corollary}\label{col:pre}
    Let $\{\ket{\psi_{S,f}}\}$ be the ensemble of $d$-dimensional random phased subspace states, instantiated with a $4$-wise independent function $f:\zo^n\to\zo$. Then the following properties hold:
    \begin{itemize}
    \item The ensemble is indistinguishable from a Haar random ensemble against:
    \begin{itemize}
        \item $\QNC^0$ circuits, when $d=\omega(\log n)$;
        \item $\AC^0\circ\QNC^0$ circuits with $\polylog(n)$ ancillae, when $d=\log^{\omega(1)} n$;
    \end{itemize}
    \item  The von Neumann entanglement entropy across any cut is at most $d$.
    \end{itemize}
\end{corollary}

\begin{proof}
First notice that even when $f$ is $4$-wise independent (or in general, $2t$-wise independent) instead of truly random, the proof of the design property of $\{\ket{\psi_{S,f}}\}$ in \Cref{sec:appendixA} still holds. Specifically, in \Cref{eq:append}:
\[
    \Ex[f]{\ket{\psi_{S,f}}\bra{\psi_{S,f}}^{\otimes t}}
    =\frac{1}{2^{dt}}\sum_{\substack{x_1,\ldots,x_t\in S\\y_1,\ldots,y_t\in S}}\Ex[f]{(-1)^{f(x_1)+\ldots+f(x_t)+f(y_1)+\ldots+f(y_t)}}\ket{x_1\cdots x_t}\bra{y_1\cdots y_t},
\]
When $f$ is $2t$-wise independent, the expectation of $(-1)^{f(x_1)+\ldots+f(x_t)+f(y_1)+\ldots+f(y_t)}$ is the same as if $f$ is truly random. As a result, $\ket{x_1\cdots x_t}\bra{y_1\cdots y_t}$ still has non-zero coefficient out only when each element of $S$ appears even number of times in $(x_1,\ldots,x_t,y_1,\ldots,y_t)$. The rest of the proof is exactly the same as in \Cref{sec:appendixA}. Specifically for $t=2$, since $\{\ket{\psi_{S,f}}\}$ is an $O(2^{-d})$-approximate $2$-design, the indistinguishability from Haar random states follows from \Cref{thm:prs} and \Cref{thm:acqnc2}.

On the other hand, note that the states are given by
\begin{equation}
\label{equation: Schmidt decomposition}
    \frac{1}{2^{d/2}}\sum_{x\in S}(-1)^{f(x)}\ket{x},
\end{equation}
where $d$ is the dimension of the subspace. For any cut of the qubits, \eqref{equation: Schmidt decomposition} gives a tensor product decomposition of the state with at most $2^d$ terms. Hence, by \Cref{schmidt rank}, the Schmidt rank of the state is at most $2^d$ and the corresponding von Neumann entanglement entropy is at most $d$.
\end{proof}

\paragraph{Preparing phased subspace states.}
Subspace states are known to be efficiently preparable  with $\mathcal{O}(nd)$ gates \cite{kerenidis2022quantummachinelearningsubspace}. The $4$-wise independent function $f:\zo^n\to\zo$ can be constructed from seeds of length $O(n)$ in $\poly(n)$ time (see e.g. \cite[Section 3.5]{vadhan2012pseudorandomness}). The phases are put into the state using an efficient controlled operation.

\section{Parallel-query secure PRU against local \texorpdfstring{$\QNC^0$}{QNC0}}\label{subsection:pru}

In previous sections we examined the pseudorandom properties of state designs against shallow quantum circuits. In this section we turn to unitary designs and show that they are also pseudorandom when queried in parallel, but against the more restricted class of circuits that are \emph{geometrically} local. Specifically, we consider the distinguisher to be a circuit $C'\cdot (U^{\otimes t}\otimes \Id)\cdot C$, where $C,C'$ are both 1-dimensional geometrically local $\QNC^0$ circuits, and $U$ is either a unitary 2-design or a Haar random unitary applied on $n$ consecutive qubits. The following folklore fact about geometrically local $\QNC^0$ circuits is crucial for us:
\begin{lemma}\label{lma:SD}
    Let $C$ be an 1-dimensional local $\QNC$ circuit of depth $d$ on $n$ qubits. Then for every $k\in[n-1]$, the Schmidt rank of the state $C\ket{0^n}$ between the first $k$ qubits and the remaining $(n-k)$ qubits is at most $4^d$.
\end{lemma}
\begin{proof}
    We prove this by an induction over $d$, and the base case when $d=0$ is trivial. Suppose that after the first $d$ layer of gates we have the Schmidt decomposition $\sum_{i=1}^{4^d}\alpha_i\ket{v_i}\ket{w_i}$, where $\ket{v_i}$ is on $k$ qubits and $\ket{w_i}$ is on $(n-k)$ qubits. In layer $d+1$, only the gate $U$ (if it exists) that acts on the $k$-th and the $(k+1)$-th qubits would affect the Schmidt rank. We can write $U$ with an arbitrary tensor product decomposition
    $U=\sum_{j=1}^4 A_j\otimes B_j$ where $A_j,B_j\in\mathbb{C}^{2\times 2}$, so that the state after applying $U$ becomes
    \[\sum_{i=1}^{4^d}\sum_{j=1}^4\alpha_i (A_j\ket{v_i})\otimes(B_j\ket{w_i}).\]
    By \Cref{schmidt rank} we know that the above state has Schmidt rank at most $4^{d+1}$, and it is not affected by the remaining gates in the same layer. 
\end{proof}
We will make use of a stronger statement that allows us to perform the Schmidt decomposition recursively on the $t$ blocks of $n$ qubits (a notion that we borrow from \cite{GharibianK12}):
\begin{lemma}[Recursive Schmidt Decomposition]\label{lma:RSD}
    Let $C$ be a 1-dimensional local $\QNC$ circuit of depth $d$ on $tn$ qubits. Then we can write the state $C\ket{0^{tn}}$ as
    \begin{equation}\label{eq:RSD}
    C\ket{0^{tn}}=\sum_{i_1,\ldots,i_t=1}^{r}\alpha_{i_1,\ldots,i_t}\cdot \ket{v_{1,i_1}}\otimes\ket{v_{2,i_1,i_2}}\otimes\cdots\otimes\ket{v_{t,i_1,\ldots,i_t}}
    \end{equation}
    where the Schmidt rank $r\leq 4^d$, and $\ket{v_{\tau,i_1,\ldots,i_\tau}}$ is an $n$-qubit state. For every $\tau\in[t]$, $i_1,\ldots,i_\tau\in[r]$ and $i_\tau'\neq i_\tau$, we have the orthogonality condition 
    \[\braket{v_{\tau,i_1,\ldots,i_{\tau-1},i_\tau}|v_{\tau,i_1,\ldots,i_{\tau-1},i_\tau'}}=0,\]
    and the complex numbers $\alpha_{i_1,\ldots,i_t}$ satisfy $\sum|\alpha_{i_1,\ldots,i_t}|^2=1$. 
\end{lemma}
\begin{proof}
    Let $r$ be the maximum Schmidt rank between the first $\tau n$ and the remaining $(t-\tau)n$ qubits, for any $\tau\in[t]$. By \Cref{lma:SD} we have $r\leq 4^d$. The recursive application of the Schmidt decomposition starts with the cut between the first $n$ qubits and the remaining $(t-1)n$ qubits:
    \[C\ket{0^{tn}}=\sum_{i_1=1}^{r}\alpha_{i_1}\cdot\ket{v_{1,i_1}}\otimes\ket{w_{1,i_1}}.\]
    
    The next step is to perform a Schmidt decomposition over each $\ket{w_{1,i_1}}$, for which we need an upper bound on the Schmidt rank. We can also write $\ket{w_{1,i_1}}$ in the computational basis to get
    \[C\ket{0^{tn}}=\sum_{i_1=1}^{r}\sum_{x\in\zo^n}\sum_{y\in\zo^{(t-2)n}}\alpha_{i_1}\beta_{i_1,x,y}\cdot\ket{v_{1,i_1}}\ket{x}\ket{y}.\]
    Since $\{\ket{v_{1,i_1}}\}$ can be expanded into an orthonormal basis on the first $n$ qubits, so can $\{\ket{v_{1,i_1}}\ket{x}\}$ on the first $2n$ qubits. As a result, the Schmidt rank of $C\ket{0^{tn}}$ between the first $2n$ and the remaining $(t-2)n$ qubits is exactly the rank of the $2^nr\times2^{(t-2)n}$ matrix $M$, where
    \[M((i_1,x),y)=\alpha_{i_1}\beta_{i_1,x,y}.\]
    On the other hand, whenever $\alpha_{i_1}\neq 0$, the Schmidt rank of $\ket{w_{1,i_1}}$ on the same cut is the rank of the submatrix of $\alpha_{i_1}^{-1}M$, with the row index $i_1$ fixed, and thus is at most $r$. Therefore we get
    \[C\ket{0^{tn}}=\sum_{i_1,i_2=1}^{r}\alpha_{i_1,i_2}\cdot \ket{v_{1,i_1}}\otimes\ket{v_{2,i_1,i_2}}\otimes\ket{w_{2,i_1,i_2}}.\]
    Continuing the process for every $\tau\in[t]$ results in a tree-like structure as in \eqref{eq:RSD}, and the sum of squares of the coefficients is guaranteed to be $1$ by the orthogonality of the states.
\end{proof}

We also need the following lemma in analogy to \Cref{lma:page}, but on cross (off-diagonal) terms conjugated with Haar random unitaries.
\begin{lemma}\label{lma:page_orth}
    Let $U$ be an $n$-qubit Haar random unitary, and let $A$ be a subsystem with $|A|=k$ qubits and $B=[n]\setminus A$. For every two $n$-qubit states $\ket{v}$ and $\ket{w}$ such that $\braket{v|w}=0$, we have
    \[\Ex{\Norm{\Tr_B[U\ket{v}\bra{w}U^\dagger]}_2^2}\leq 2^{k-n}.\]
\end{lemma}

\begin{proof}
Without loss of generality, we can assume that $U\ket{w}=\ket{0^n}$, and $U\ket{v}=\ket{u}$ is a uniformly random state orthogonal to $\ket{0^n}$. In this case, we can write
\begin{align*}
    \Tr_B[U\ket{v}\bra{w}U^\dagger] &= \sum_{x\in\zo^{n-k}}(\Id_A\otimes \ket{x}\bra{x}_B)\ket{u}\bra{0^n}(\Id_A\otimes \ket{x}\bra{x}_B)\\
    &= (\Id_A\otimes \ket{0^{n-k}}\bra{0^{n-k}}_B)\ket{u}\bra{0^n}.
\end{align*}
And thus,
\begin{align*}
    \Norm{\Tr_B[U\ket{v}\bra{w}U^\dagger]}_2^2 &= \Tr\Brac{(\Id_A\otimes \ket{0^{n-k}}\bra{0^{n-k}}_B)\ket{u}\braket{0^n|0^n}\bra{u}(\Id_A\otimes \ket{0^{n-k}}\bra{0^{n-k}}_B)} \\
    &= \Tr\Brac{(\Id_A\otimes \ket{0^{n-k}}\bra{0^{n-k}}_B)\ket{u}\bra{u}} \\
    &= \sum_{x\in\zo^k}\Abs{\braket{x_A0^{n-k}_B|u}}^2.
\end{align*}
As a side note, this directly implies that
\begin{equation}\label{eq:tr_frob}
    \Norm{\Tr_B[U\ket{v}\bra{w}U^\dagger]}_2\leq 1
\end{equation}
regardless of the choice of $\ket{v}$, $\ket{w}$ or $U$, which will be useful later on.

Since $\ket{u}$ is a uniformly random state orthogonal to $\ket{0^n}$, for every $x\in\zo^n\setminus\{0^n\}$, $\braket{x|u}$ has the same distribution. Therefore every term in the summation above, except $\Abs{\braket{0^n|u}}^2$, has the same expectation which is $1/(2^n-1)$. Therefore we conclude that
\[\Ex{\Norm{\Tr_B[U\ket{v}\bra{w}U^\dagger]}_2^2} = \frac{2^k-1}{2^n-1}< 2^{k-n}. \qedhere\]
\end{proof}

The above result allows us to show that the output of non-adaptive queries of a Haar random unitary has the property that every subsystem of $k$ qubits are almost maximally mixed, similar to \Cref{col:page_haar}, when the input state admits the recursive Schmidt decomposition.

\begin{corollary}\label{col:pru_haar}
    Let $U$ be an $n$-qubit Haar random unitary. For $t\geq 1$, let $\ket{\psi}$ be a $tn$-qubit state that admits the recursive Schmidt decomposition with rank $r$ as in \Cref{lma:RSD}. Let $\rho_A$ be the reduced density matrix of $U^{\otimes t}\ket{\psi}\bra{\psi}U^{\dagger\otimes t}$ over a subsystem $A$ with $|A|=k$ qubits. Then for every $\delta>0$, with probability at least $1- O(r)\cdot2^{k-n/2}\cdot\delta^{-1}$ over $U$, it holds that
    \[\Norm{\rho_A-\frac{1}{2^k}\Id_k}_1\leq \delta.\]
\end{corollary}
\begin{proof}
    To take the partial trace for subsystem $A$, we assume that $A$ consists of $k_1,\ldots,k_t$ qubits in each block of $n$ qubits respectively, and let $B_\tau$ be the part outside $A$ in the $\tau$-th block. With the recursive Schmidt decomposition  \eqref{eq:RSD} we can write 
    \begin{equation}\label{eq:pru}
    U^{\otimes t}\ket{\psi}\bra{\psi}U^{\dagger\otimes t} = \sum_{\substack{i_1,\ldots,i_t,\\j_1,\ldots,j_t=1}}^{r}
    \alpha_{i_1,\ldots,i_t}\overline{\alpha_{j_1,\ldots,j_t}}\cdot U\ket{v_{1,i_1}}\bra{v_{1,j_1}}U^\dagger\otimes\cdots\otimes
    U\ket{v_{t,i_1,\ldots,i_t}}\bra{v_{t,j_1,\ldots,j_t}}U^\dagger.
    \end{equation}
    In addition, we can also write the decomposition \eqref{eq:RSD} only to a certain level $\tau<t$ to get
    \[\ket{\psi}=\sum_{i_1,\ldots,i_\tau=1}^{r}\alpha_{i_1,\ldots,i_\tau}\cdot \ket{v_{1,i_1}}\otimes\ket{v_{2,i_1,i_2}}\otimes\cdots\otimes\ket{v_{\tau,i_1,\ldots,i_\tau}}\otimes\ket{w_{\tau,i_1,\ldots,i_\tau}},\]
    where $\ket{w_{\tau,i_1,\ldots,i_\tau}}$ is a state on $(t-\tau)n$ qubits such that
    \[\alpha_{i_1,\ldots,i_\tau}\ket{w_{\tau,i_1,\ldots,i_\tau}}=\sum_{i_{\tau+1},\ldots,i_t=1}^r \alpha_{i_1,\ldots,i_t}\cdot\ket{v_{\tau+1,i_1,\ldots,i_{\tau+1}}}\otimes\cdots\otimes\ket{v_{t,i_1,\ldots,i_t}}.\]
    Notice that it also implies
    \[|\alpha_{i_1,\ldots,i_\tau}|^2=\sum_{i_{\tau+1},\ldots,i_t=1}^r|\alpha_{i_1,\ldots,i_t}|^2.\]
    This way, we can group the summands in \eqref{eq:pru} depending on the smallest coordinate $\tau$ such that $i_\tau\neq j_\tau$ (when $\tau=t+1$, it means that $(i_1,\ldots,i_t)$ is identical to $(j_1,\ldots,j_t)$), and have
    \begin{align}
    U^{\otimes t}\ket{\psi}\bra{\psi}U^{\dagger\otimes t} = &\  \sum_{\tau=1}^{t+1}\sum_{i_1,\ldots,i_\tau=1}^r\sum_{j_\tau\neq i_\tau}
    \alpha_{i_1,\ldots,i_\tau}\overline{\alpha_{i_1,\ldots,i_{\tau-1},j_\tau}}\cdot U\ket{v_{1,i_1}}\bra{v_{1,i_1}}U^\dagger\otimes\cdots \nonumber\\
    &\ \otimes U\ket{v_{\tau,i_1,\ldots,i_\tau}}\bra{v_{\tau,i_1,\ldots,i_{\tau-1},j_\tau}}U^\dagger\otimes
    U^{\otimes(t-\tau)}\ket{w_{\tau,i_1,\ldots,i_\tau}}\bra{w_{\tau,i_1,\ldots,i_{\tau-1},j_\tau}}U^{\dagger\otimes(t-\tau)}. \label{eq:pru2}
    \end{align}
    
    Now consider each summand in \eqref{eq:pru2} with $\tau\leq t$. By \Cref{lma:page_orth}, we have
    \[\Ex{\Norm{\Tr_{B_\tau}[U\ket{v_{\tau,i_1,\ldots,i_\tau}}\bra{v_{\tau,i_1,\ldots,i_{\tau-1},j_\tau}}U^\dagger]}_2^2} \leq 2^{k_\tau-n}.\]
    Therefore, after taking the partial trace, with the bound \eqref{eq:tr_frob} on the Frobenius norms of the other blocks, we can bound the expected Frobenius norm of the entire summand by $\Abs{\alpha_{i_1,\ldots,i_\tau}\overline{\alpha_{i_1,\ldots,i_{\tau-1},j_\tau}}}\cdot 2^{(k_{\tau}-n)/2}$. Thus by linearity of expectation and triangular inequality, these summands in total has expected Frobenius norm of at most
    \begin{align*}
    &\ \sum_{\tau=1}^{t}\sum_{i_1,\ldots,i_\tau=1}^r\sum_{j_\tau\neq i_\tau}
    \Abs{\alpha_{i_1,\ldots,i_\tau}\overline{\alpha_{i_1,\ldots,i_{\tau-1},j_\tau}}}\cdot 2^{(k_\tau-n)/2} \\
    = &\ \sum_{\tau=1}^{t}\sum_{i_1,\ldots,i_{\tau-1}=1}^r\left(\sum_{i_\tau=1}^r\Abs{\alpha_{i_1,\ldots,i_\tau}}\right)^2\cdot 2^{(k_\tau-n)/2} \\
    \leq &\ \sum_{\tau=1}^t\sum_{i_1,\ldots,i_\tau=1}^r\Abs{\alpha_{i_1,\ldots,i_\tau}}^2\cdot r\cdot2^{(k_\tau-n)/2} \\
    \leq &\ r\cdot 2^{(k-n)/2}.
    \end{align*}

    The remaining summands are those with $(i_1,\ldots,i_t)=(j_1,\ldots,j_t)$, and their sum is exactly
    \[\sum_{i_1,\ldots,i_t=1}^r\Abs{\alpha_{i_1,\ldots,i_t}}^2\cdot U\ket{v_{1,i_1}}\bra{v_{1,i_1}}U^\dagger\otimes\cdots\otimes
    U\ket{v_{t,i_1,\ldots,i_t}}\bra{v_{t,i_1,\ldots,i_t}}U^\dagger.\]
    By \Cref{lma:page}, similarly to the deduction in \Cref{col:page_haar}, we know that the partial trace in each block, denoted by
    $M_\tau=\Tr_{B_\tau}[U\ket{v_{\tau,i_1,\ldots,i_\tau}}\bra{v_{\tau,i_1,\ldots,i_\tau}}U^\dagger]$,satisfies that
    \[\Ex{\Norm{M_\tau-\frac{1}{2^{k_\tau}}\Id_{k_\tau}}_2}\leq 2^{(k_\tau-n)/2}.\]
    Thus by a hybrid argument, we have
    \begin{align*}
        \Ex{\Norm{M_1\otimes\cdots\otimes M_t-\frac{1}{2^k}\Id_k}_2} &
        \leq \sum_{\tau=1}^t\Ex{\Norm{M_1\otimes\cdots\otimes M_{\tau-1}\otimes\left(M_\tau-\frac{1}{2^{k_\tau}}\Id_{k_\tau}\right)}_2}
        \\
        &\leq \sum_{\tau=1}^t 2^{(k_\tau-n)/2}
        \leq2^{(k-n)/2}.
    \end{align*}
    
    Since $\sum\Abs{\alpha_{i_1,\ldots,i_t}}^2=1$, we conclude that
    \[\Ex{\Norm{\rho_A-\frac{1}{2^k}\Id_k}_1}\leq 2^{k/2}\cdot \Ex{\Norm{\rho_A-\frac{1}{2^k}\Id_k}_2}\leq (r+1)\cdot 2^{k-n/2}.\]
    By Markov's inequality, we know that $\Norm{\rho_A-\frac{1}{2^k}\Id_k}_1\leq\delta$ holds with probability at least $1-(r+1)\cdot2^{k-n/2}\cdot\delta^{-1}$.
\end{proof}
Since the proof of \Cref{col:pru_haar} only uses the second moment properties of $U$ (\Cref{lma:page} and \Cref{lma:page_orth} to be exact), using \Cref{lma:approxdesign} we can also conclude the following for approximate unitary 2-designs.
\begin{corollary}\label{col:pru_design}
    Let $\{U_i\}$ be an $\eps$-approximate unitary $2$-design on $n$ qubits. For $t\geq 1$, let $\ket{\psi}$ be a $tn$-qubit state that admits the recursive Schmidt decomposition with rank $r$ as in \Cref{lma:RSD}. Let $\rho_A$ be the reduced density matrix of $U_i^{\otimes t}\ket{\psi}\bra{\psi}U_i^{\dagger\otimes t}$ over a subsystem $A$ with $|A|=k$ qubits. Then for every $\delta>0$, with probability at least $1-2^{O(k)}\cdot r\cdot(\eps+2^{-n})^{1/2}\cdot\delta^{-1}$ over $i$, it holds that
    \[\Norm{\rho_A-\frac{1}{2^k}\Id_k}_1\leq \delta.\]
\end{corollary}

Combining the above corollaries together, we obtain the desired pseudorandomness against geometrically local $\QNC^0$ circuits.
\begin{theorem}
\label{thm:qncpru}
    Let $\{U_i\}$ be an $\eps$-approximate unitary $2$-design on $n$ qubits for $\eps=2^{-\Omega(n)}$. Then $\{U_i\}$ is non-adaptive secure PRU against $1$-dimensional geometrically local $\QNC$ circuits of depth $d=\log n-\omega(1)$. Moreover, the pre-processing part of the $\QNC$ circuit (before applying $U_i$) could have depth up to $o(n)$.
\end{theorem}
\begin{proof}
    The proof follows from that of \Cref{thm:prs}. Since the output of the $\QNC$ circuit depends on only $k=2^d=o(n)$ qubits from the output of $U^{\otimes t}$, by \Cref{col:pru_haar} and \Cref{col:pru_design} we know that the outputs have $\negl(n)$ trace distance between the Haar random $U$ and unitary design $\{U_i\}$, with probability $1-\negl(n)$ over the choice of the unitaries. This is because
    \[2^{O(k)}\cdot r\cdot(\eps+2^{-n})^{1/2}\leq\negl(n)\]
    when $k=o(n)$ and $\eps=2^{-\Omega(n)}$, and $r\leq 2^d$ from \Cref{lma:RSD}. In addition, the circuit depth $d$ used to bound the Schmidt rank $r$ only concerns the depth of the pre-processing part of the circuit, which can be raised up to $o(n)$.
\end{proof}

Although \Cref{col:pru_haar} and \Cref{col:pru_design} have virtually the same form as \Cref{col:page_haar} and \Cref{col:page_design}, we cannot use the techniques in \Cref{subsection: AC0} to directly obtain the similar security of PRU against $\QNC^0$ circuits with $\AC^0$ post-processing. The reason is that in \Cref{subsection: AC0} we have to limit the number of ancillae, and here they correspond to the qubits that $U$ is not applied to, which we cannot put any natural restrictions on. However, in the artificial scenario where we force the the unitary to be applied in parallel on all qubits (and therefore allow no ancillae for the post-processing $\AC^0\circ\QNC^0$ circuit), we can obtain the following statement in analogy to \Cref{thm:acqnc}:
\begin{theorem}
\label{thm:acqncpru}
    Let $\{U_i\}$ be an $\eps$-approximate unitary $2$-design on $n$ qubits for $\eps=2^{-\log^{\omega(1)}n}$. Then $\{U_i\}$ is PRU against a subclass of $\AC^0\circ\QNC^0$ circuits on a multiple of $n$ qubits, where $U_i$ is applied non-adaptively over all the qubits used by the circuit, and the pre-processing $\QNC^0$ circuit before applying $U_i$ is $1$-dimensional geometrically local. 
\end{theorem}

\section{Discussion and outlook}
\label{sec: discussions}
Our work initiates the study of unconditionally fooling shallow quantum circuits. In general, statistical and computational notions of quantum pseudorandomness---such as $t$-designs and PRS---are incomparable. Our work shows that in the low-complexity regime, the two notions can in fact overlap and illustrate rich connections to computational complexity theory, reminiscent of the intimate relation between hardness and randomness in classical computation. Our work leaves a number of interesting questions regarding the connections between hardness and quantum pseudorandomness. We discuss a few of these below.

\paragraph{Fooling stronger circuits.}
What are the strongest class of quantum circuits that $t$-designs, or in particular $2$-designs, can fool? We conjecture that $2$-designs are computationally secure against $\QAC^0$ circuits as well. To prove this, it suffices to show a $\QAC$ analogy of Braverman's result on \emph{almost $k$-wise maximally mixed} input states. These states have the property that every subsystem on $k$ qubits is close to being maximally mixed, and we conjecture that they cannot be distinguished from the true maximally mixed state by $\QAC$ circuits of small depths.

A more immediate improvement on our results would be the removal of the constraints on ancillae in \Cref{thm:acqnc2} and \Cref{col:acqnc}. The reason we require a bounded number of ancillae is purely technical: we need this to argue that the output distribution has high min-entropy, as otherwise the $k$-wise indistinguishability would not guarantee that we can fool $\AC^0$ circuits. We note that a similar techinical issue occured in the $\AC^0\circ\QNC^0$ lower bound lower bound result of \cite{Slote24}.

\paragraph{Stronger security for PRU.}
The unconditional security we proved for PRU is quite limited. Specifically, in \Cref{thm:qncpru} we could only show security when the unitaries are queried non-adaptively, while the adversaries are 1-dimensional geometrically local $\QNC$ circuits. Can we lift the requirement of non-adaptivity or geometric locality?

We conjecture that new constructions are necessary in order to achieve security against adaptive queries. In other words, there exist (approximate) $t$-designs that are not PRU against $\QNC^0$ circuits with adaptive queries. Such an example that works for arbitrary $t\leq\poly(n)$ would also give a separation between non-adaptive and adaptive PRU.

\paragraph{Optimal pseudoentanglement.}
It is known that the optimal entanglement entropy gap for pseudoentanglement is $\omega(\log n)$ versus $O(n)$, which is achievable across every cut when assuming the existence of post-quantum one way functions \cite{aaronson2023quantumpseudoentanglement}. In \Cref{col:pre} we showed explicit examples of unconditional pseudoentanglement, which is optimal across every cut against $\QNC^0$ distinguishers, but only $\log^{\omega(1)}n$ versus $O(n)$ against an $\AC^0\circ\QNC^0$ adversary. Can we also construct unconditional optimal pseudoentanglement against $\AC^0\circ\QNC^0$, or even stronger circuits?

\paragraph{PRG against shallow circuits.}
In this work we did not consider classical pseudorandom primitives, such as PRGs. One of the reasons is that we need to be careful about the definition when the adversaries are shallow quantum circuits, since it makes a difference whether or not we allow access to multiple copies of the PRG output. In fact, if the adversary could only access a single copy, then the classical Nisan-Wigderson generator \cite{Nisan92,NisanW94} instantiated by the parity function would directly give an unconditional PRG with $\polylog (n)$ seed length against $\QAC^0$ circuits with bounded number of ancillae, using the recently results on the hardness of parity against $\QAC^0$ \cite{Nadimpalli2024QAC0,anshu2024Qac0Superlinear}. However, the hardness proofs in these works, which examine the Pauli spectrum of the $\QAC^0$ circuit, break down when allowing multi-copy access to the input, as even a single classical fan-out gate could significantly increase the Pauli weight of the overall circuit.

As a result, designing unconditional multi-copy secure PRG against $\QAC^0$ circuits remains to be an intriguing open problem. A follow-up direction is to use such PRG to construct other pseudorandom primitives such as pseudorandom functions or even PRS and PRU. Notice that we have recipes for these constructions against polynomial-sized adversaries, but in order to work against bounded adversaries, the construction needs to be super-efficient and such constructions are largely unknown.

\paragraph{Fooling other models.}
Finally, an open-ended question is to find unconditional pseudorandom objects against other restricted models of quantum computation, with or without our constructions and techniques. It was shown in \cite{GirishR22} that the INW generator \cite{ImpagliazzoNW94} is secure against space-bounded quantum computation. To study PRS or PRU against such models, the first challenge lies in finding the most relevant and useful definition, which we leave for future work.

\subsection*{Acknowledgements}
The authors would like to thank John Bostanci, Daniel Grier, Natalie Parham, Jack Morris, Francisca Vasconcelos, and Henry Yuen for stimulating discussions on related topics. SS is supported by a Royal Commission for the Exhibition of 1851 Research Fellowship.
\printbibliography  

\appendix

\section{Design properties of random phased subspace states}
\label{sec:appendixA}
Here we show that the random phased subspace states $\ket{\psi_{S,f}}$, defined in \Cref{def:subspaces}, form an approximate design. Fixing the subspace $S$ and taking the average over the random function $f:S\to\zo$, we have
\begin{equation}\label{eq:append}
    \Ex[f]{\ket{\psi_{S,f}}\bra{\psi_{S,f}}^{\otimes t}}
    =\frac{1}{2^{dt}}\sum_{\substack{x_1,\ldots,x_t\in S\\y_1,\ldots,y_t\in S}}\Ex[f]{(-1)^{f(x_1)+\ldots+f(x_t)+f(y_1)+\ldots+f(y_t)}}\ket{x_1\cdots x_t}\bra{y_1\cdots y_t},
\end{equation}
where the term $\ket{x_1\cdots x_t}\bra{y_1\cdots y_t}$ has non-zero coefficient only when each element of $S$ appears an even number of times in $(x_1,\ldots,x_t,y_1,\ldots,y_t)$. Among those let us consider the ones such that $x_1,\ldots,x_t$ are all distinct (so that $y_1,\ldots,y_t$ is a permutation of $x_1,\ldots,x_t$); the partial summation over these terms is
\begin{align*}
    &\ \frac{1}{2^{dt}}\sum_{\substack{x_1,\ldots,x_t\in S\\x_i\neq x_j,\forall i<j\\\pi\in \cS_t}} \ket{x_1\cdots x_t}\bra{x_{\pi(1)}\cdots x_{\pi(t)}} \\
    =&\ \frac{1}{2^{dt}}\sum_{\{x_1,\ldots,x_t\}\subset S}\Paren{\sum_{\pi\in\cS_t}\ket{x_{\pi(1)}\cdots x_{\pi(t)}}}\Paren{\sum_{\pi\in\cS_t}\bra{x_{\pi(1)}\cdots x_{\pi(t)}}} \\
    =&\ \frac{t!}{2^{dt}}\sum_{X\subset S,|X|=t}\ket{\mathsf{Sym_X}}\bra{\mathsf{Sym}_X}.
\end{align*}
Here $\cS_t$ is the symmetric group on $[t]$, and $\ket{\mathsf{Sym_X}}=\frac{1}{\sqrt{t!}}\sum_{\pi\in\cS_t}\ket{x_{\pi(1)}\cdots x_{\pi(t)}}$ for $X=\{x_1,\ldots,x_t\}$. As a result, the above partial sum has trace
\[\frac{t!}{2^{dt}}\cdot\binom{2^d}{t}=\frac{2^d\cdot(2^d-1)\cdot\cdots\cdot(2^d-t+1)}{2^{dt}}\geq 1-\frac{t^2}{2^d}.\]
Also notice that the partial sum can be thought as the projection of $\Ex[f]{\ket{\psi_{S,f}}\bra{\psi_{S,f}}^{\otimes t}}$ onto the subspace spanned by $\{\ket{\mathsf{Sym}_X}\}_{X\subset S,|X|=t}$, implying that the remaining part is still positive-definite. This allows us to bound the trace distance:
\begin{align*}
    &\ \Norm{\Ex[f]{\ket{\psi_{S,f}}\bra{\psi_{S,f}}^{\otimes t}}-\binom{2^d}{t}^{-1}\sum_{X\subset S,|X|=t}\ket{\mathsf{Sym_X}}\bra{\mathsf{Sym}_X}}_1 \\
    \leq&\ \Norm{\Ex[f]{\ket{\psi_{S,f}}\bra{\psi_{S,f}}^{\otimes t}}-\frac{t!}{2^{dt}}\sum_{X\subset S,|X|=t}\ket{\mathsf{Sym_X}}\bra{\mathsf{Sym}_X}}_1+\Abs{\frac{t!}{2^{dt}}\cdot\binom{2^d}{t}-1}
    \leq\frac{2t^2}{2^d}.
\end{align*}

Now we think of $S$ as a uniformly random $d$-dimensional subspace of $\zo^n$. For a uniformly random $X\subset S$ with $|X|=t$, the elements in $X$ are linearly dependent with probability at most
\[\frac{1}{2^d}+\frac{1}{2^{d-1}}+\cdots+\frac{1}{2^{d-t+1}}<\frac{1}{2^{d-t}}.\]
Similarly, when $X$ is a uniformly random size-$t$ subset of $\zo^n$, the elements in $X$ are linearly dependent with probability at most $2^{t-n}$. When conditioned on linear independence, the distributions of $X$ in both cases are the same, and thus
\[\Norm{\Ex[S]{\binom{2^d}{t}^{-1}\sum_{X\subset S,|X|=t}\ket{\mathsf{Sym_X}}\bra{\mathsf{Sym}_X}}-\binom{2^n}{t}^{-1}\sum_{X\subset \zo^n,|X|=t}\ket{\mathsf{Sym_X}}\bra{\mathsf{Sym}_X}}_1\leq 2^{t-d}+2^{t-n}.\]
Hence we conclude that
\[\Norm{\Ex[S,f]{\ket{\psi_{S,f}}\bra{\psi_{S,f}}^{\otimes t}}-\binom{2^n}{t}^{-1}\sum_{X\subset \zo^n,|X|=t}\ket{\mathsf{Sym_X}}\bra{\mathsf{Sym}_X}}_1\leq \frac{2t^2}{2^d}+2^{t-d}+2^{t-n}=O(2^{t-d}).\]

On the other hand, for the Haar random state $\ket{\psi_\haar}$, it is well known that (see e.g. \cite{mele2024introduction})
\[\Ex{\ket{\psi_\haar}\bra{\psi_\haar}^{\otimes t}}=\binom{2^n+t-1}{t}^{-1}\Pi_{\mathsf{sym}}^{(t,2^n)},\]
where $\Pi_{\mathsf{sym}}^{(t,2^n)}$ is the projection onto the symmetric subspace of $(\mathbb{C}^{2^n})^{\otimes t}$. Since all the $\ket{\mathsf{Sym}_X}$ take up $\binom{2^n}{t}$ dimensions in the subspace, their weight in $\Pi_{\mathsf{sym}}^{(t,2^n)}$ is at least
\[\binom{2^n}{t}\Big/\binom{2^n+t-1}{t}=\frac{2^n\cdot(2^n-1)\cdot\cdots\cdot(2^n-t+1)}{2^n\cdot(2^n+1)\cdot\cdots\cdot(2^n+t-1)}\geq 1-\frac{t^2}{2^n}.\]
This means that
\[\Norm{\Ex{\ket{\psi_\haar}\bra{\psi_\haar}^{\otimes t}}-\binom{2^n}{t}^{-1}\sum_{X\subset \zo^n,|X|=t}\ket{\mathsf{Sym_X}}\bra{\mathsf{Sym}_X}}_1\leq \frac{2t^2}{2^n},\]
and we can finally obtain that
\[\Norm{\Ex[S,f]{\ket{\psi_{S,f}}\bra{\psi_{S,f}}^{\otimes t}}-\Ex{\ket{\psi_\haar}\bra{\psi_\haar}^{\otimes t}}}_1\leq \frac{2t^2}{2^n}+O(2^{t-d})\leq O(2^{t-d}).\]

\end{document}